\documentclass[a4paper,twocolumn,11pt,accepted=2025-09-22]{quantumarticle}
\pdfoutput=1
\usepackage{dcolumn}
\usepackage[utf8]{inputenc}
\usepackage{graphicx,bm,color}
\usepackage[dvipsnames]{xcolor}
\usepackage[margin=1in]{geometry}
\usepackage{amsthm,amsmath,amssymb,mathtools}
\usepackage{multirow}
\usepackage{todonotes}
\usepackage{tikz-cd} 
\usepackage{braket}
\usepackage[normalem]{ulem}
\definecolor{britishracinggreen}{rgb}{0.0, 0.5, 0.15}
\definecolor{burgundy}{rgb}{0.5, 0.0, 0.13}
\definecolor{egyptianblue}{rgb}{0.06, 0.2, 0.65}
\usepackage[colorlinks=true, citecolor=burgundy, linkcolor=britishracinggreen, urlcolor=egyptianblue]{hyperref}
\usepackage{tikz-cd}
\usepackage{enumitem}
\usepackage{ragged2e}
\usepackage[noend]{algpseudocode}
\usepackage{float}
\makeatletter
\let\newfloat\newfloat@ltx
\makeatother
\usepackage{algorithm}
\usepackage{multirow}

\usepackage{comment}

\newcommand{\ketbra}[2]{\ket{#1}\!\bra{#2}}

\newcommand{\reals}{{\mathbb R}}

\usepackage{fbm_usra_macros}
\graphicspath{{./figures/}}

\begin{document}
\title{Improving Quantum Approximate Optimization by Noise-Directed Adaptive Remapping}
\author{Filip B. Maciejewski}
\email{fmaciejewski@usra.edu}
\author{Jacob Biamonte}
\author{Stuart Hadfield}
\author{Davide Venturelli}
\email{davide.venturelli@nasa.gov}
\affiliation{\quail}
\affiliation{\riacs}
\date{2025.11.01}
\begin{abstract}
\begin{minipage}{1.0\linewidth}
We present Noise-Directed Adaptive Remapping (NDAR), a heuristic algorithm for approximately solving binary optimization problems by leveraging certain types of noise.
We consider access to a noisy quantum processor with dynamics that features a global attractor state.
In a standard setting, such noise can be detrimental to the quantum optimization performance.
Our algorithm bootstraps the noise attractor state by iteratively gauge-transforming the cost-function Hamiltonian in a way that transforms the noise attractor into higher-quality solutions.
The transformation effectively changes the attractor into a higher-quality solution of the Hamiltonian based on the results of the previous step.
The end result is that noise aids variational optimization, as opposed to hindering it.
We present an improved Quantum Approximate Optimization Algorithm (QAOA) runs in experiments on Rigetti's quantum device.
We report approximation ratios 0.9--0.96 for random, fully connected graphs on \(n=82\) qubits, using only depth \(p=1\) QAOA with NDAR.
This compares to 0.34--0.51 for standard \(p=1\) QAOA with the same number of function calls. 
\end{minipage}
\end{abstract}

\begin{figure*}[!t]
    \centering
\includegraphics[width=0.98\textwidth]{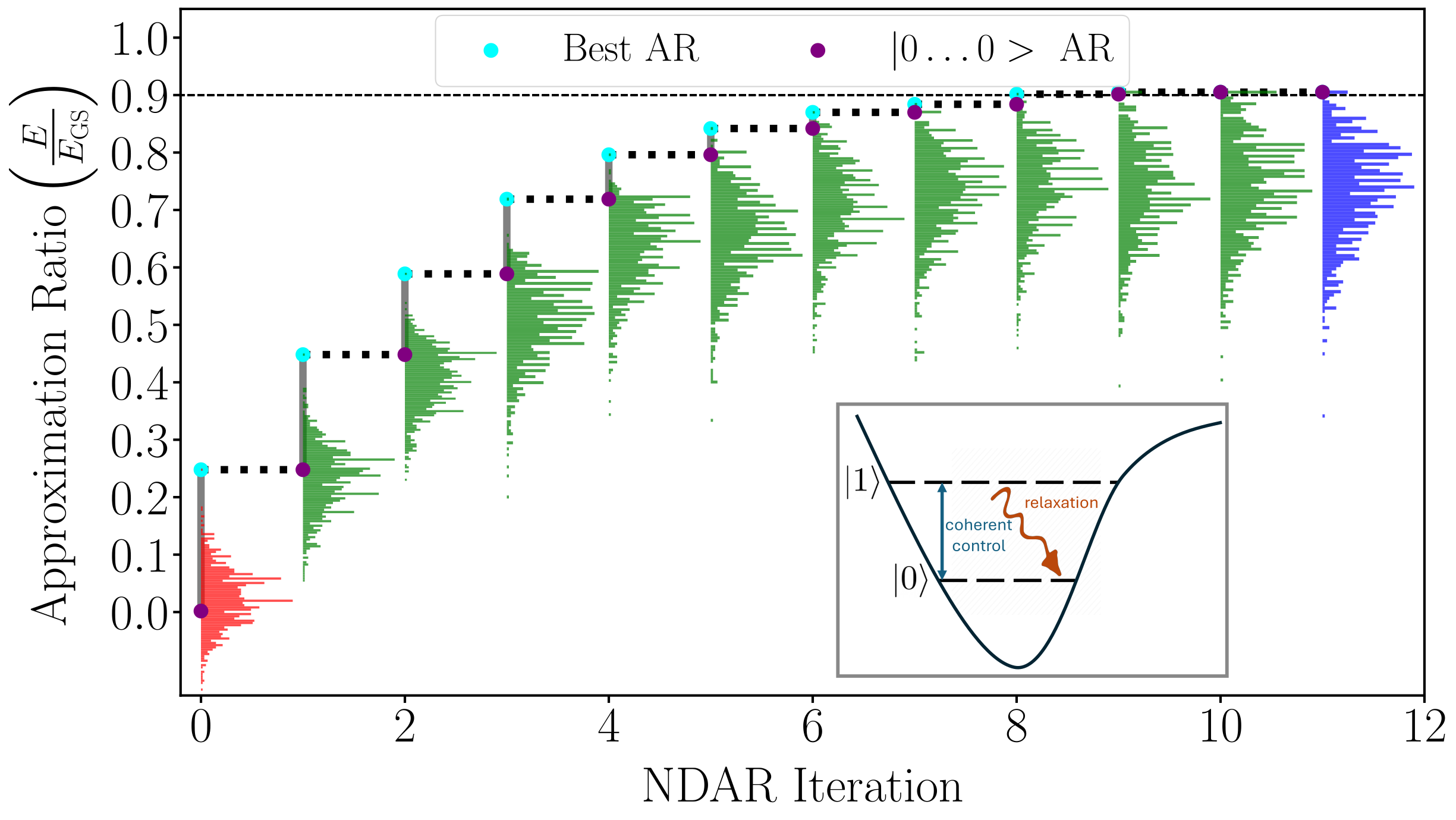}
    \caption{
 Illustration of a Noise-Directed Adaptive Remapping run for a single Hamiltonian instance.
 The X-axis represents iterations of the algorithm.
 The purple dots correspond to the approximation ratio (AR) of the $\ket{0\dots 0}$ (attractor) state, while cyan dots are the best-observed solution (as a result of QAOA runs).
 In consecutive steps, the Hamiltonian is re-transformed in a way that the best solution from the previous step effectively becomes the (approximate) attractor of the new step, resulting in a stairs-like climb towards higher approximation ratios until convergence.
 The distributions of ARs are stacked horizontally for each NDAR step, showing how the tail of the distribution from a previous step, becomes a point near the center of the distribution in the next step (in early iterations).
 The presented data corresponds to a single instance from the experimentally obtained dataset (for $n=82$ qubits) in Fig.~\ref{fig:exp:NDAR}b.
 The inset illustrates how binary (spin) variables are associated with the qubit states of a physical QPU system. Qubits are represented by anharmonic oscillators (e.g.~transmons) with ground state $\ket{0}$ and first excited state $\ket{1}$. Gates drive coherent transitions (blue arrows) while noise (red arrow) induces asymmetric incoherent transitions.} \label{fig:ndar_stairs_illustration}
\end{figure*}
 
\section{Introduction}

Noise in state-of-the-art quantum processors limits quantum circuit depth. 
Significant effort has been spent on improving the implementation of quantum approximate optimization~\cite{farhi2014quantum,hadfield2019quantum,abbas2023quantum,sachdeva2024quantumoptimizationusing127qubit,dupont2023quantum, Zhu2022AdaptQAOAOriginal, Dupont2024QRR,maciejewski2023design,dupont2024lightcones,Ebadi2022MIS,Byun2022MIS,Kim2022MIS,King2023dwave5000spins,Nguyen2023rydberg}, and developing various error mitigation~\cite{cai2023quantum} and error correction~\cite{knill1997TheoryOfQEC,roffe2019QECTutorial} techniques.
In this work, we propose an algorithmic framework that does not mitigate the noise directly, but instead attempts to align the quantum optimization with the dissipative part of the evolution by using the knowledge about the noise.

To this aim, we exploit the concept of bitflip transforms: problem remappings 
that exchange the logical definitions of a set of binary variables, or equivalently that exchange the interpretation of $\ket{0}$ and $\ket{1}$ basis states for a set of physical qubits. 
Assuming that the device noise dynamics is biasing outputs towards a specific classical ``attractor" state (denoted as $\ket{0\dots0}$), we propose an algorithm that iteratively gauge-transforms the problem in a way that consistently assigns better cost-function values to the noise attractor.
The remappings are adaptively selected based on the best solution obtained through variational optimization in the previous step (see Fig.~\ref{fig:ndar_stairs_illustration} for illustration). 

We refer to the new algorithm as Noise-Directed Adaptive Remapping (NDAR).  
The actual variational optimization method is a subroutine of NDAR that can be chosen in a way best suited for a given experimental setup (in fact, it does not necessarily have to be quantum).
In this work, we focus mainly on the canonical example of Quantum Approximate Optimization Algorithm (QAOA)~\cite{farhi2014quantum,hadfield2019quantum} though the adaptive remapping approach we introduce applies to other approaches.
We implement the NDAR adaptive loop for $p=1$ QAOA targeting the minimization of fully-connected, randomly-weighted Sherrington-Kirkpatrick Hamiltonians on $\noq=82$-qubit subsystem of Rigetti Computing's quantum processing unit (QPU) based on superconducting transmons.
We observe solutions with approximation ratios above $\approx 0.9$, compared to $\approx 0.5$ of standard QAOA, when both methods are given the same number of circuit runs.

\paragraph*{Related Work} The use of bitflip gauges (a.k.a.~spin-reversal or spin-inversion transforms) has been considered for quantum annealing in the context of embedding problem instances to ``unbias" individual qubits and generally to reduce performance errors~\cite{boixo2013experimental,king2014algorithm,perdomo2015performance,perdomo2016determination,pudenz2016parameter,pelofske2019optimizing,barbosa2021optimizing,izquierdo2021ferromagnetically,DWaveSRT}. 
In practice, bitflip transforms are often chosen at random, see, e.g., ~\cite{boixo2013experimental,pelofske2019optimizing,DWaveSRT}.
In Ref.~\cite{pelofske2019optimizing}, the authors developed a heuristic algorithm that searches for the best gauge, and presented experimental results on a quantum annealer, treating the bitflip gauge as an additional parameter to be tuned.
Our approach is conceptually different -- the NDAR choice of 
bitflip transforms is informed by the noise model, and it is iteratively updated in the course of implementation.
For quantum gate model applications, some studies investigate and potentially exploit various problem symmetries (such as qubit ordering);
see, for example, Refs.~\cite{cirqQP,qiskitNoise,murali2019noise,ji2023improving,matsuo2023sat, shaydulin2021classical, maciejewski2023design,galda2021transferability,shaydulin2021error}. 
Effects of various types of noise on QAOA performance are studied, for example, in Refs.~\cite{alam2019analysis,xue2021effects,marshall2020characterizing,Maciejewski2021modeling}.
The idea of adaptively changing the ansatz is found in some recent works, such as the ADAPT-QAOA~\cite{Zhu2022AdaptQAOAOriginal} protocol which adaptively constructs the mixer Hamiltonian.
In the case of NDAR applied to QAOA, applying 
a bitflip transform is equivalent to adaptively changing the physically implemented phase separation 
(cost) 
Hamiltonians, and adding a simple post-processing procedure (re-labeling). 
Generally speaking, we note that Noise-Directed Adaptive Remapping might fit into a more general framework of noise-assisted quantum algorithms that use open systems dynamics as a resource~\cite{verstraete2008quantumcomputationquantumstate,Kastoryano2011dissipative,Reiter2017dissipative,foldager2021noiseassisted,Cao2021,Leppkangas2023,cubitt2023DQE, ding2023SingleAncilla, kastoryano2023quantum, mi2024stable, chen2023LocalMinima}.
Finally, we note that the NDAR algorithm 
is conceptually related but distinct from 
so-called ``warm start" strategies for quantum approximate optimization \cite{egger2021warm,tate2023bridging, tate2023warm, wurtz2021classically}, 
where the ansatz or initial state are modified based on approximate solutions obtained classically. 
In the case of NDAR, the ansatz is adaptively modified based on the approximate solutions obtained in \emph{quantum} optimization (and based on our knowledge of noise dynamics). 
Finally, we emphasize that NDAR is applicable to a wide variety of quantum and classical algorithms, especially those categorized as Ising machines~\cite{mohseni2022ising}.

\section{Problem Remappings}

A large variety of industrially relevant optimization tasks can be formulated as optimization problems over binary variables~\cite{ausiello2012complexity,hadfield2019quantum}.
For example, numerous important problems~\cite{lucas2014ising,hadfield2021representation} can be directly translated or reduced to quadratic unconstrained optimization problems over 
Boolean variables $x_i=0,1$, or, equivalently, instances of the  
Ising model:
    $\min\sum_{i} h_{i} s_i + \sum_{i<j} J_{ij} s_is_j$
for classical spin variables $s_i=\pm1$, and real coefficients~$h_i,J_{ij}$.
While generalizations to non-Ising cost functions for NDAR are possible, for the rest of the paper we will suppose we are 
given such a problem cost function we wish to minimize, along with a suitable quantum stochastic computational method for sampling from possible problem instance solutions.

For quantum optimization methods, such as those based on QAOA and generalizations, variable values are typically mapped to qubit basis states~$\ket{0},\ket{1}$. 
The cost function is then  mapped~\cite{lucas2014ising,hadfield2021representation} to a diagonal cost Hamiltonian 
\begin{align}\label{eq:cost_hamiltonian}
    H = \sum_{i} h_{i} Z_i + \sum_{i<j} J_{ij} Z_iZ_j +\dots
\end{align}
where $Z_i$ is Pauli $Z$ operator acting on qubit $i$, and $J_{i,j} \in \mathbb{R}$ is interaction strength between qubits $i,j$, while $h_i$ is a local field magnitude.

\paragraph{Bitflip transforms} Consider the unitary ``bitflip" operator $P_{\y} = \bigotimes_{i=0}^{\noq-1}X_{i}^{y_i}$ that acts to flip the $|0\rangle$, $|1\rangle$ basis states as specified by the bitstring $y\in\{0,1\}^\noq$.
This change-of-basis can be applied to $H$ 
by incorporating the change of signs to the weights of $h_i$ and $J_{i,j}$ as 
$H \rightarrow H^{\y}$ with 
\begin{align}\label{eq:gauge_transformation_local}
    H^{\y} 
    &= P_{\y}HP_{\y}\\
    &= \sum_{i} \left(-1\right)^{y_i}h_{i} Z_i 
    + \sum_{i<j} \left(-1\right)^{y_i+y_j}J_{i,j} Z_iZ_j +\dots. \nonumber
\end{align}
This transformation preserves cost Hamiltonian eigenvalues, with eigenvectors (candidate problem solutions) permuted~under $P_{\y}$. 
In particular, under this transformation the $|0\dots0\rangle$ state is mapped to~$|y_0\dots y_{n-1}\rangle$.

\paragraph{Gauge-transformed QAOA circuits}\label{sec:preliminaries:qaoa_symmetries}
Consider a standard QAOA setup, where the quantum circuit 
consists of alternating layers of phase and mixer operators applied to the initial state~$\ket{+}^{\otimes \noq}$, followed by measurement in the computational basis.  
The phase and mixer operators at layer $l$ are given by $U_P(\gamma_l)=exp(-i\gamma_lH)$ and $U_M(\beta_l)=exp(-i\beta_l\sum^\noq_j X_j)$, respectively. The parameters $\gamma_l$ and $\beta_l$ control the algorithm performance and they are assumed to be set by means of a \emph{parameter setting strategy} ~\cite{neira2024benchmarking}.
The physical implementation of QAOA implicitly assumes a particular choice of mapping of the problem variables and $H$. 
In the noiseless setting, QAOA circuits are robust to the choice of basis in the sense that 
bitflip transformations of the cost Hamiltonian $H \rightarrow H^{\y}$ are equivalent to transforming the entire circuit $U \rightarrow U^{\y}$.
This follows easily from the properties $U_P(H^{\y})=P_{\y}U_P(H)P_{\y}$, $P_{\y}U_MP_{\y}=U_M$,
and $P_{\y}\ket{+}^{\otimes \noq}=\ket{+}^{\otimes \noq}$
as shown in Appendix~\ref{app:qaoa_symmetries}. 
Hence all resulting probabilities of measuring particular outcomes (via Born's rule) remain the same, up to a known bitflip transformation, and so cost Hamiltonian expectation values are invariant.  
Nevertheless, in general, implementing 
$U_P(H^{\y})$ gives different unitaries for different $\y$.

\paragraph{Noise breaks bitflip gauge symmetry}

Clearly, we do not expect such bitflip gauge invariance of QAOA to hold in the presence of noise,  
since generally the non-controlled part of the evolution has no a priori reason to respect a given symmetry. 
Consider a simple example of the independent amplitude-damping channel.
In the Kraus operator representation, this corresponds to a probabilistic application of the transition operator $\ketbra{0}{1}$ to 
each qubit. 
Since $\ketbra{0}{1}$ does not commute with bitflips, the QAOA circuit probabilities become dependent on the particular problem mapping and are no longer invariant under bitflip transformations.
Hence QAOA's performance will depend on the particular choice of $\y$ for the cost Hamiltonian $H^{\y}$, even though different $H^{\y}$ all encode the same problem. 
One way of interpreting this example is that applying a bitflip gauge transformation effectively changes the stationary state of the noisy dynamics (with respect to the problem to be solved), and thus can potentially improve the algorithm by aligning the attractor state with high-quality solutions of the cost function.
For example, it is intuitively clear that in the case of amplitude damping, the states featuring low Hamming weights (more 0s than 1s) are favored by noisy evolution.
This suggests that problem mappings for which low-Hamming-weight states have low cost might be beneficial.
Indeed, for example, for the well-studied Maximum Cut problem, the state $\ket{00\dots 0}$ corresponds to the worst possible cut (i.e., the empty set) in the standard formulation. 
If we instead, say, bitflipped the mapping of half the variables, then $\ket{00\dots 0}$ would instead correspond to a cut 
near the random guessing expectation value (i.e., half of the edges), 
a dramatic improvement over zero. 
We exploit these observations in the following sections to improve QAOA runs through NDAR, where we utilize an adaptive approach to account for the noise present in the Ankaa-2 QPU.

\section{Noise-Directed Adaptive Remapping Algorithm}

Inspired by the above observations, we propose the Noise-Directed Adaptive Remapping (NDAR) summarized as Algorithm~\ref{alg:NDAR_main} (and illustrated in Fig.~\ref{fig:ndar_stairs_illustration}).
The NDAR consists of an external loop, in which each iteration involves implementing a stochastic optimization (such as QAOA) for a gauge-transformed Hamiltonian. 
The gauge transformation in each step is chosen based on the results of the optimization in the previous iteration -- it is defined by the best-cost bitstring obtained so far. 
The NDAR  ensures that the (postulated) steady state of the noise in each iteration step is the lowest energy state observed so far (see Fig.~\ref{fig:ndar_stairs_illustration}).
This
adaptively changes the optimized cost function in a way that greedily improves the energy of the attractor state, potentially leading to improved optimization.

In Algorithm~\ref{alg:NDAR_main},  we leave a generic $\texttt{stochastic optimizer}$ that encompasses some black-box optimization subroutine. 
In each NDAR iteration, the $\texttt{stochastic optimizer}$ takes $H^{\y}$ and a number of allowed samples $M$ as an input, and returns a list of $M$ proposed solutions to the problem $H^{\y}$ (represented conventionally as bitstrings).
This can correspond, for example, to a standard QAOA circuit runs coupled with an arbitrary parameter setting strategy.
We also allow a general $\texttt{TERMINATION RULE}$ to denote a stopping condition for the NDAR external loop.
Various heuristics can be used -- for example, terminate if the cost value did not decrease sufficiently relative to the previous iteration steps~\cite{vinci2016optimally}.
Note that in the Algorithm~\ref{alg:NDAR_main}, we use $\ket{0\dots 0}$ as the attractor state. The algorithm can be easily modified for a different attractor which doesn't have to be known a priori but could also be discovered experimentally (for example, by running identity circuits via techniques similar to unitary folding \cite{giurgica2020digital}).

\begin{algorithm}
\caption{\centering Noise-Directed Adaptive Remapping (Greedy version, attractor $\equiv|0\dots0\rangle$)}\label{alg:NDAR_main}
\RaggedRight\textbf{Input}:\\
\RaggedRight $H_{\left(0\right)}$: Hamiltonian encoding the cost function\\ 
 \RaggedRight $\texttt{stochastic optimizer}$: a sub-algorithm that performs a stochastic optimization
 of $H_{\left(l\right)}$ (gauge-transformed Hamiltonian in iteration step $l+1$). 
 Example: standard (noisy) QAOA ansatz.\\
\RaggedRight $M$: total number of samples in each optimization step \\
\RaggedRight $\texttt{TERMINATION RULE}$: a rule specifying termination of the loop. Example: the energy did not decrease relative to the previous iteration.
\vspace{0.1cm}\\
 \RaggedRight\textbf{Pseudocode for
 iteration $j\geq 1$}:\\ \vspace{-5pt}
\begin{enumerate}[wide, labelwidth=0pt, leftmargin=12pt, align=left]
\item Run optimization for Hamiltonian $H_{\left(j-1\right)}$ using $\texttt{stochastic optimizer}$.
The result is $M$ bitstrings (measurement outcomes) $\cbracket{\x^i_{\left(j\right)}}_{i=1}^{M}$.
\item Calculate cost values $\cbracket{E_{\left(j\right)}^{i}= \bra{\x^i_{\left(j\right)}}H_{\left(j-1\right)}\ket{\x^i_{\left(j\right)}}}_{i=1}^{M}$ 
corresponding to the Hamiltonian $H_{\left(j-1\right)}$.
\item Select the bitstring ${\x^{\ast}}$ corresponding to the minimum of the cost function $E_{\left(j\right)}^{\ast}=\min_{i} E_{\left(j\right)}^i$. \\
\textbf{if} \texttt{TERMINATION RULE}:\\ 
\hspace{25pt}\textbf{return} lowest energy solution found overall\\
\textbf{else}:\\
\hspace{25pt}proceed to step 4.
\item
Apply bitflip gauge transform $P_{\x^{\ast}}$ to $H_{(j-1)}$,
obtaining new Hamiltonian $H_{\left(j\right)} \coloneqq H_{\left(j-1\right)}^{\x^{\ast}}$ (cf. Eq.~\eqref{eq:gauge_transformation_local}). 
This ensures that the transformed Hamiltonian 
has the property $\bra{00\dots 0}H_{\left(j\right)}\ket{00\dots 0} = E^{\ast}_{\left(j\right)}$.
Note that $H_{\left(j\right)}$ is equivalent to the original cost Hamiltonian $H_{\left(0\right)}$ via the product of gauge transformations applied so far, and so the solution to the original problem is easily recovered. 
Set $j=j+1$. 
Proceed to step 1.
\end{enumerate}
\end{algorithm}

\begin{figure*}[!t]
    \centering
\includegraphics[width=0.98\textwidth]{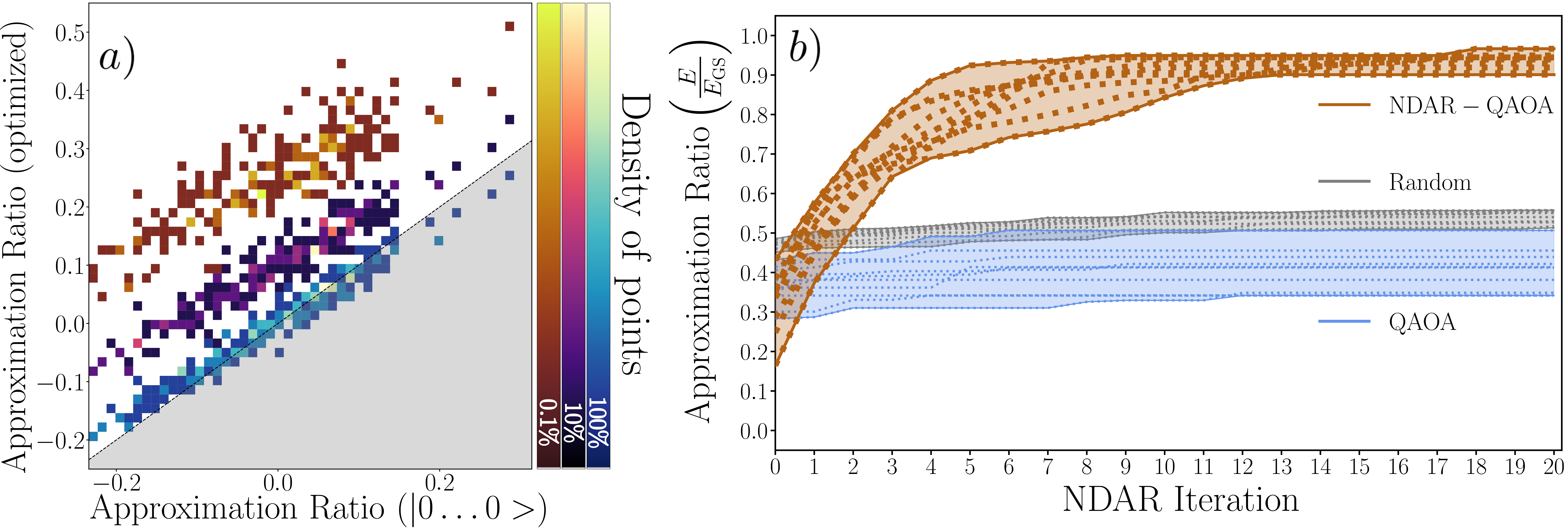}
\caption{
Investigation of the effects of applying
bitflip gauges to improve an $n=82$-qubit $p=1$ QAOA in experiments on Rigetti's Ankaa-2 superconducting quantum device.
a) Correlation between the approximation ratio (AR) of the attractor state $\ket{0\dots0}$ (X-axis) and the AR attained via QAOA (Y-axis) using transformed Hamiltonian.
Differently colored regions of the histogram correspond to the AR estimators constructed from different percentiles of the AR  distribution.
We computed Pearson’s correlation $r_q$ and Spearman’s rank correlation $\rho_q$ coefficients at each percentile $q$.
We obtained
$r_{0.1\%}=0.87^{+0.03}_{-0.04}$, $r_{10\%}=0.964^{+0.009}_{-0.012}$, and $r_{100\%}=0.984^{+0.004}_{-0.005}$; and
$\rho_{0.1\%}=0.86^{+0.04}_{-0.05}$, $\rho_{10\%}=0.96^{+0.01}_{-0.02}$, and $\rho_{100\%}=0.981^{+0.006}_{-0.007}$.
Confidence intervals are 95\% bias-corrected and accelerated ($BC_{a}$) paired-bootstrap intervals~\cite{efron1994introductionBootstrap} with $10^5$ resamples, calculated using the standard \texttt{scipy.stats} module \cite{scipy}. Two-sided tests yielded p-values $<10^{-58}$ for all cases.
Taken together, these results indicate a strong, nearly linear association between the AR of the $\ket{0\cdots 0}$ state and the quality of optimization.
The shaded region corresponds to the $y\leq x$ where the optimized AR is lower than the AR of the attractor state (indicating no performance gain).
Data points are aggregated for $10$ Hamiltonian instances.
b) Overview of Noise-Directed Adaptive Remapping algorithm (see Algorithm~\ref{alg:NDAR_main}) performance as a function of iterations (X-axis).   
The Y-axis shows the AR of the best-found solution (out of $1000*150*(i+1)$, with $i$ being the iteration index). 
Each dotted line corresponds to a different random Hamiltonian instance, solid lines indicate the maximal and minimal performance across $10$ instances.  
Besides best-found solutions obtained via NDAR applied to QAOA (orange), we plot standard QAOA optimization (light blue), and a baseline that corresponds to uniform random sampling (grey).
Since in the case of QAOA and random baseline the protocols are not iterative, the iteration index corresponds to increasing the number of cost function evaluations/gathered samples to match those of NDAR.
}\label{fig:exp:NDAR}
\end{figure*}

\section{Experimental demonstration}

In this section, we present the results of experiments performed on 82-qubit subsystems of Rigetti's superconducting transmon device Ankaa-2. 
Technical details of implementation are provided in Appendix~\ref{app:experiments_details}.
Furthermore, in Appendix~\ref{app:simulations_details} we report numerical demonstrations of NDAR (with simulated amplitude damping) for small system sizes.

We start by empirically investigating how the choice of the bitflip gauge affects the results of variational optimization.
We define the approximation ratio (AR) as the value of the cost function applied on a specific output (i.e., a bitstring) of the QAOA circuit after measurement in the computational basis, divided by the best-known cost of the problem ($E_{GS}$). Multiple runs result in multiple outputs which define an AR distribution.

We conjecture that QAOA optimization will perform better for bitflip gauges that transform a Hamiltonian in such a way that the approximation ratio (AR) of the attractor state $\bra{0\dots 0}H^{\y}\ket{0,\dots 0}/{E_{GS}}$ (or equivalently $\bra{\y}H\ket{\y}/{E_{GS}}$) is higher. 
Hence we first investigate a correlation between said AR and the performance of the QAOA on the cost Hamiltonian after running the $\y$-transformed circuit with the chosen parameter setting strategy.
The Hamiltonian cost function we optimize corresponds to fully connected, weighted graphs 
\begin{align}\label{eq:SK_hamiltonian}
    H = \sum_{i<j} J_{i, j} Z_iZ_j ,
\end{align}
with 2-body interaction coefficients $J_{i,j}$ drawn uniformly at random from the bimodal distribution of values $\pm1$ (so-called Sherrington-Kirkpatrick (SK) model \cite{sherrington1975solvable}).

We implement $p=1$ QAOA optimization for $n=82$ qubits with Tree-Structured Parzen Estimators (TPE) \cite{hyperopt2013,optuna2019} parameter setting strategy using a budget of $t=100$ cost function evaluations (trials).
Each trial corresponds to estimating the cost function value from $s=1000$ samples at fixed parameters.
The QAOA is implemented for a set of 10 distinct Hamiltonian cost functions. 
For each Hamiltonian instance, the optimization is run for $20$ random bitflip gauges $P_{\y}$.

For each QAOA run, we collect the distribution of ARs and compare the result with $\bra{\y}H\ket{\y}/{E_{GS}}$.
We construct the estimators from three quantiles of the total distribution over samples (for the best trial), corresponding to averaging only a subset of all AR outputs~\cite{ronnow2014defining} after sorting them in descending order ($1$, $100$, and $1000$ samples, which constitute the top $0.1\%$, $10\%$, and $100\%$ of all of the output samples for the best trial, respectively).
We aggregate the data for all test Hamiltonians and plot the results as 2-dimensional histograms on the left side of Fig~\ref{fig:exp:NDAR}. 

We observe a clear positive correlation between the AR of the optimized distribution and the AR of the attractor state evaluated on the cost Hamiltonian for all quantiles, as confirmed by the estimated values of Pearson's and Spearman's correlation coefficients~\cite{zwillinger1999crc} close to $1.0$, reported in the Figure's caption.
Importantly, there is a clear separation between the ARs of best-found solutions (top $0.1\%$) and the ARs of the attractor state, as seen by the fact that all data points lie above the shadowed $x\leq y$ region. 
This variance suggests that choosing a proper bitflip gauge and collecting sufficient samples can indeed aid the algorithm performance in experiments.

We exploit this observation further by implementing the Noise-Directed Adaptive Remapping algorithm applied to find solutions to the same class of $n=82$-qubit Hamiltonians with the same ansatz ($p=1$ QAOA).
Each step of the NDAR loop entails implementing TPE parameter search using $t=150$ trials, with $s=1000$ samples each. 
Following Ref.~\cite{maciejewski2023design}, we additionally allow the TPE loop to optimize over categorical variable controlling the gates' ordering of the ansatz (see Appendix~\ref{app:experiments_details} for details).
We specify $\texttt{TERMINATION RULE}$ for NDAR (recall Algorithm~\ref{alg:NDAR_main}) that terminates the loop if neither the minimal found energy nor the cost function (i.e., mean energy estimator) decreases compared to the previous step. 

The results are presented on the right side of Fig.~\ref{fig:exp:NDAR}.
The data points correspond to the best-found solution, meaning the best bitstring out of $150*1000*(i+1)$ samples ($i$ being iteration step counting from $0$). 
The NDAR QAOA results from QPU are colored orange.
We also implement the standard QAOA (light blue) with the same number of function calls as NDAR. 
For reference, we provide the best approximation ratios obtained from the same number of shots using uniformly random sampling (grey).

We observe that Noise-Directed Adaptive Remapping improves the best-found solutions from range $\approx \left[0.34,0.51\right]$ (variation between Hamiltonians for standard QAOA, blue lines) to $\approx \left[0.9, 0.96\right]$.
In particular, without NDAR, standard QAOA (light blue) performs worse than random sampling (grey) when we consider the best sample, while NDAR (orange) allows us to outperform both the standard QAOA and random sampling already after 3 iterations.
Importantly, the NDAR algorithm with QAOA converged after at most $20$ optimization runs (achieving high-quality results much earlier), suggesting that the overhead from using NDAR is not prohibitively high for $n=82$ qubits. Note that we did not fine-tune the meta-parameters of the TPE loop, and it is likely that both performance and convergence can be improved.

\begin{figure*}[!t]
    \centering
\includegraphics[width=0.98\textwidth]{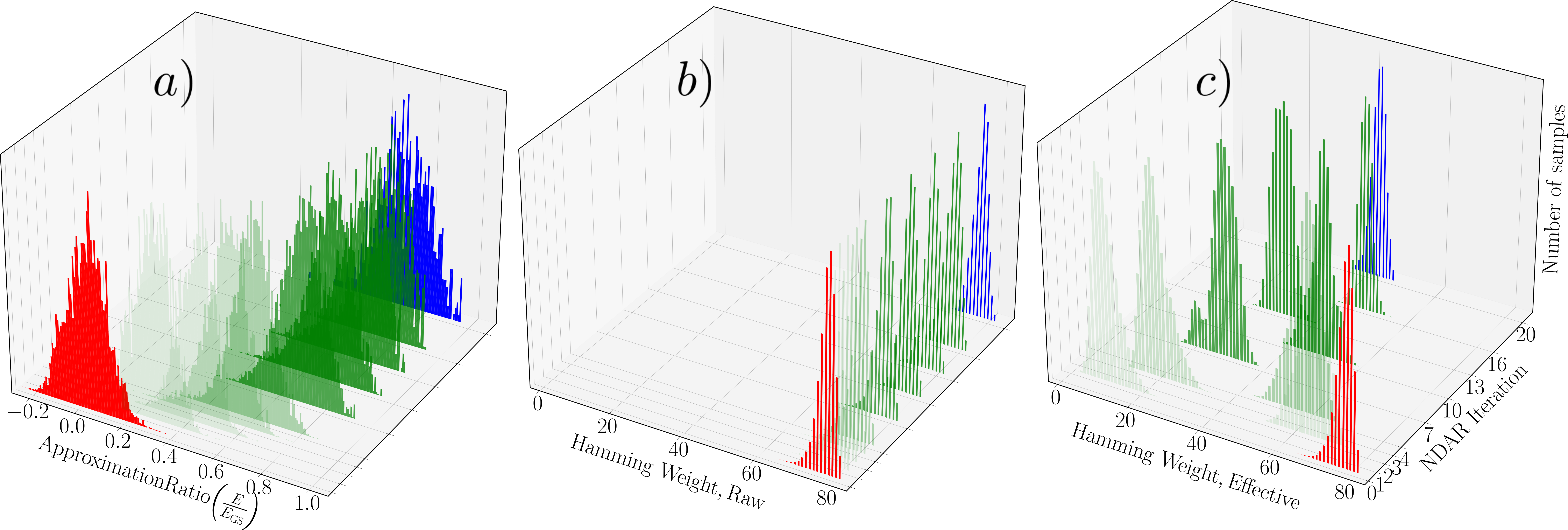}
\caption{\label{fig:exp:NDAR_histograms}
Evolution of the samples' distributions properties as a function of NDAR iteration.
The Y-axis is NDAR iteration, and the Z-axis represents a (normalized) number of occurrences.
The X-axis corresponds to the a) Approximation Ratio, b) the physically recorded, raw Hamming weight, and c) the effective Hamming weight attained due to NDAR bitstrings relabeling.
First and last iterations are colored red and blue, respectively; while intermediate iterations are graduated green.
The plots show aggregated data for $10$ Hamiltonians.
For each Hamiltonian, there is a total of $1000$ samples shown for each iteration (taken from the optimized trials from the same data as in Fig.~\ref{fig:exp:NDAR}b).
}
\end{figure*}
In Fig.~\ref{fig:exp:NDAR_histograms}a, we present a 3-dimensional visualization of how the AR distributions evolve as a function of NDAR iteration (aggregating all instances, differently from Fig.~\ref{fig:ndar_stairs_illustration}).  
We observe how the distributions converge towards higher approximation ratio regions, following the increasing quality of the solution corresponding to the attractor state.
One of the important effects of noise on QAOA optimization is that it highly contracts the attainable search space (for example, in the case of strong amplitude damping, this corresponds to restricting the search to only low hamming weights basis states).
At the same time, we expect that NDAR, due to the relabeling of basis elements, will effectively expand that space (by effectively remapping the Hamming sectors), allowing to concentrate the search in regions where better solutions lie.

To investigate that effect, in Figs.~\ref{fig:exp:NDAR_histograms}b and ~\ref{fig:exp:NDAR_histograms}c, we plot an evolution of two types of Hamming weight (HW) distributions.
The first type (middle plot) corresponds to physically registered measurement outcomes ('raw' data) that are highly concentrated close to an attractor state (which experimentally appears to be closer to $\ket{1\dots 1}$ rather than $\ket{0 \dots 0}$, a fact we attribute to possible overheating of the chip, see also Appendix~\ref{app:experiments_details:optimization}).
This demonstrates how physical noise restricts the search space and explains low approximation ratios attained in Fig.~\ref{fig:exp:NDAR}b, even if QAOA (without NDAR) is given many cost function evaluations.
The second type (right plot) corresponds to the \emph{effective} (or remapped) Hamming weight that results from gauge transformations applied to the measurement basis elements. 
The samples distributions are initially concentrated close to very high HW (initial optimization, iteration $0$, red), and then exhibit oscillatory behavior with increasing iterations, to finally converge around $\frac{n}{2}=41$ Hamming weights (last iteration, blue).
Note that for generic random SK Hamiltonians, we expect the ground states to have Hamming weight close to $\frac{n}{2}$, thus the NDAR convergence to that region is a desirable behavior. 

\section{Conclusions and future work}

We introduced Noise-Directed Adaptive Remapping (NDAR), a new algorithmic framework for quantum approximate optimization in the presence of certain types of noise.
NDAR iteratively changes the encoding of the target Hamiltonian in a way that adaptively attempts to align the quantum evolution with the noise dynamics.
We believe that the Noise-Directed Adaptive Remapping algorithm has the potential to become a standard method in the general toolbox of quantum optimization in the non-fault-tolerant regime, and, more generally, a technique to exploit the interplay of symmetries and open-system dynamics to improve algorithm performance.
In particular, as many superconducting, photonic, and electronic quantum information processing platforms are affected by amplitude-damping-like noise, we expect NDAR to be particularly useful for optimization implemented in those devices (including quantum annealers and analog simulators).
At the same time, it is important to emphasize the limitations of our method. It relies heavily on assumptions about the noise model—specifically, that it features a classical attractor state. Such a model does not typically apply to certain platforms, such as ion traps or neutral atoms, which can make the vanilla NDAR impractical. An interesting open question is how to generalize NDAR (for example, by using more general gauge transformations) so that it can be applied to such hardware. Alternatively, one could artificially introduce noise into the protocol to enforce additional bias, thereby allowing NDAR to exploit knowledge of previously found solutions (see Ref.~\cite{tam2025enhancingNDAR} for a promising study in this direction).

In the previous sections, we empirically demonstrated the effectiveness of NDAR in a proof-of-concept $p=1$ standard QAOA optimization on Rigetti's $82$-qubit QPU,
observing large performance gains over vanilla QAOA without NDAR. It is worth stressing that those gains are \emph{not} due to noise-mitigation effects (as we are not getting closer to the ideal QAOA), but apparently come from the feedback-loop structure of the protocol -- in fact, we suspect that noiseless QAOA with $p=1$ could perform worse than the presented NDAR implementation. 
While the presented experimental results seem promising, more analytical and numerical work is required to better understand the general applicability of our method, as well as its dependence on the noise type and strength.
Moreover, our study unveils numerous avenues for further research, a few of which we now outline.

First, we note that the proposed formulation of NDAR is intentionally very basic, and many heuristics can be applied to extend it.
For example, the outer NDAR loop is currently greedy, meaning we always choose the new bitflip gauge based on the best-found solution.
This has clear downsides, as follows from studies of classical greedy algorithms.
It is possible that more flexible strategies would be beneficial -- such as allowing for optimization over multiple low-energy gauges, performing the choice of the bitflip gauge based on its Hamming weight or typicality, or performing an additional local (classical) search in the neighborhood of the best-found solutions.

Second, we tailored the presentation on the Ising model, which is very general in its utility and has the convenience of supporting an easy remapping of the Hamiltonian via a simple change of the signs of its coefficients. However, similar simple transformations exist in other more complex models (such as the q-state Potts model that can be related to constrained optimization of the coloring type~\cite{inaba2022potts}). In general, NDAR could be formulated for polynomial cost functions of integer variables if an appropriate basis is defined, which could also include qudits encodings~\cite{bravyi2022hybrid,ozguler2022numerical}. The study of this rich space of mapping transformations and Hamiltonian encodings is likely a very exciting route of investigation.

Third, we believe it would be interesting to implement NDAR with different algorithms than QAOA. For example, NDAR might be conceptually related to recently studied dissipative quantum algorithms that aim to mimic the cooling processes in Nature to find low-energy states, see Refs.~\cite{cubitt2023DQE, ding2023SingleAncilla, kastoryano2023quantum, mi2024stable, chen2023LocalMinima}.
Bath-engineering techniques, which have shown promise to avoid the barren plateau problem, might also be coupled with our remappings~\cite{sannia2023engineered}.
Moreover, experimental tests with different types of unitary ansätze (including those that are not necessarily bitflip-gauge invariant, such as Quantum Mixer-Phaser Ansatz \cite{larose2022mixer}) and optimization algorithms are warranted \cite{hadfield2019quantum, Zhu2022AdaptQAOAOriginal,herrman2022multi,Wurtz2022counteradiabaticity,magann2022feedback, Liu2022layerVQE, maciejewski2023design, blekos2306review, wilkie2024qaoa,Shaydulin2024evidence,vijendran2024expressive}.
In particular, it would be intriguing to combine NDAR with methods such as ADAPT-QAOA \cite{Zhu2022AdaptQAOAOriginal}; with iterative/recursive versions of QAOA \cite{bravyi2020obstacles,dupont2023quantum,brady2023iterative};; with multilevel/multigrid quantum optimization strategies \cite{ushijima2021multilevel,angone2023hybrid,maciejewski2024multilevel,bach2024mlqaoa,acharya2024decomposition,bach2025solving};
with alternative qubit-efficient mappings/encodings~\cite{sundar2024qubit,sciorilli2025towards,sciorilli2025competitive};
 with warm-start techniques~\cite{egger2021warm,tate2023bridging, tate2023warm, wurtz2021classically} or with recently-proposed quantum optimization schemes that generate candidate solutions via expectation-values statistics~\cite{Dupont2024RelaxAndRound,dupont2025optimization}.

Finally, we should note that NDAR idea could be applied to a variety of non-quantum Ising machines~\cite{mohseni2022ising}, including digital algorithms, where the noise can be present because of hardware design, or could be added as a functional driver to exploit the attractor effect.

\section*{Code availability}
GitHub repository containing the Python code for implementation of NDAR is available under \href{https://github.com/usra-riacs/quantum-approximate-optimization}{github.com/usra-riacs/quantum-approximate-optimization}~\cite{quapopt_repo}.

\section*{Acknowledgments}
This work was supported by the Defense Advanced Research Projects Agency (DARPA) under Agreement No. HR00112090058 and IAA8839, Annex 114. 
Authors acknowledge support under NASA Academic Mission Services under contract No. NNA16BD14C.
We thank Maxime Dupont, Bram Evert, Mark Hodson, Stephen Jeffrey, Aaron Lott, Salvatore Mandrà, and Bhuvanesh Sundar for fruitful discussions at various stages of this project.

\bibliographystyle{quantum}
\bibliography{bib}

\clearpage

\onecolumngrid
\appendix
\begin{center}
{\large Supplementary Materials for\\ Improving Quantum Approximate Optimization\\ by Noise-Directed Adaptive Remapping}
\end{center}

\section{Gauge transformations of QAOA circuits}\label{app:qaoa_symmetries}

Recall an $n$-qubit QAOA circuit of depth $p$ and parameters $\mathbf{\gamma},\mathbf{\beta}\in\reals^p$ is given by
\begin{align}\label{eq:generic_qaoa_circuit}
U\left(H;\mathbf{\gamma},\mathbf{\beta}\right) := \prod_{l=p}^{1} U_{M}\left(\beta_l\right)U_{P}\left(\gamma_l\right), 
\end{align}
with phase separator 
\begin{align}
    U_{P}\left(\gamma_l\right)  := U_{P}\left(H;\gamma_l\right) = \exp\left(-i\gamma_l H\right) 
\end{align}
and mixer operator
\begin{align}
    U_{M} \left(\beta_k\right) := \bigotimes_{j=0}^{n-1} \exp\left(-i\beta_k X_j\right).
\end{align}
The QAOA operator $U(H;\mathbf{\gamma},\mathbf{\beta})$ is applied to the standard initial state 
$\ket{s}=\ket{+}^{\otimes n}$, where $\ket{+} = \frac{1}{\sqrt{2}}\left(\ket{0}+\ket{1}\right)$.
We now formalize the observation that  
bitflip gauge transformations of the input cost Hamiltonian do not affect QAOA performance in the ideal setting. The same result does not generally apply to other ans\"atze. Let $Q_{\x}:=\ket{\x}\bra{\x}$ such that 
$\bra{\psi} Q_{\x} \ket{\psi}=\langle Q_{\x}\rangle_\psi$ gives the probability of measuring bitstring $\x$ for any quantum state $\ket{\psi}$.
In what follows, we denote $Q_{\x^{\y}} =\ket{\x^\y}\bra{\x^\y} \coloneqq P_{\y} \ketbra{\x}{\x} P_{\y}$.

\begin{lem}
    Bitflip (spin-reversal) gauge transformations $H\rightarrow H^{\y}$ do not affect the performance of QAOA in the noise-free setting, i.e. for any fixed $p,\gamma,\beta$ and $\x,\y\in\{0,1\}^n$ we have 
$$\langle Q_{\x^{\y}}\rangle_{U(H^{\y};\gamma,\beta)\ket{s}}=\langle Q_{\x} \rangle_{U(H;\gamma,\beta)\ket{s}},$$
i.e. the probability of measuring the transformed bitstring $\x^{\y}$ is the same as that of measuring $\x$ originally, 
  and hence
    $$\langle H^{\y}\rangle_{U(H^{\y};\gamma,\beta)\ket{s}}=\langle H \rangle_{U(H;\gamma,\beta)\ket{s}}.$$
\end{lem}
\begin{proof}
Here we take advantage of the useful properties $P_{\y}^\dagger = P_{\y}$, $P_{\y}^2=I$, $P_{\y}X_jP_{\y}=X_j$, and $P_{\y}\ket{s}=\ket{s}$, for any  $\y\in\{0,1\}^n$. 
Let $p=1$; the arbitrary $p$ case follows similarly. Fix $\mathbf{\gamma},\mathbf{\beta}$ and observe that 
    $$U(H^{\y})=e^{-i\beta_l B} e^{-i\gamma_l P_{y}H P_{y}}
    =e^{-i\beta_l B}P_{y} e^{-i\gamma_l H }P_{y}
    =P_{y}e^{-i\beta_l B} e^{-i\gamma_l H }P_{y}
    =P_{y}U(H)P_{y},$$
and so applying the properties above gives
$$\bra{\x} U(H^\y)\ket{s}= \bra{\x} P_\y U(H)P_\y\ket{s}=\bra{\x^\y} U(H)\ket{s},$$
i.e., quantum amplitudes are preserved which implies measurement probabilities are too. 
\end{proof}
We remark that related techniques are applied to obtain results concerning the interplay of QAOA and classical cost function symmetries in \cite{shaydulin2021classical}. In the setting of that work we are given a transformation $A$ which commutes with a particular cost Hamiltonian $H$, i.e. $AH-HA=0$, and the affect on QAOA dynamics and outcomes is studied. In our distinct setting $P_{\y}$ and $H$ do not  
commute as this would imply $H^{\y}=H$.

The lemma shows that on an ideal quantum device, gauge transformations would have no effect on QAOA performance.  
In contrast, the proofs breaks down in the presence of noise. 

\begin{rem}\label{rem:gauges_change_of_basis}
Note that the above is mathematically equivalent to applying basis transformation defined by $P_{\y}$ to every unitary in the circuit, but is operationally different from it -- change of basis would correspond to applying $P_{\y}$ to every unitary, as opposed to just phase separators.
The equivalence between running QAOA for $H^{\y}$ and performing a global change of basis does not hold in the same way for arbitrary quantum algorithms. For other ans\"atze that similarly do not depend on the choice of bitflip gauge, the algorithm parameters may have to transform also if, for example, the mixer of a more general QAOA circuit~\cite{hadfield2019quantum} does not commute with the particular $P_{\y}$. 
Other cases of parameterized quantum circuits may depend on the particular ordering (permutation) of variables in the chosen encoding, such as 
Quantum Alternating Mixer-Phaser Ansatz \cite{larose2022mixer}. 
Clearly, NDAR can be applied in such scenarios as well. 
In particular, one can imagine a black-box algorithm where the circuit does not depend on the target problem.
Then the NDAR feedback loop would not modify the circuit at all, but only the measured gauge-transformed Hamiltonian.
\end{rem}

\section{Numerical studies}\label{app:simulations_details}

\begin{figure*}[!h]
    \centering
\includegraphics[width=0.55\textwidth]{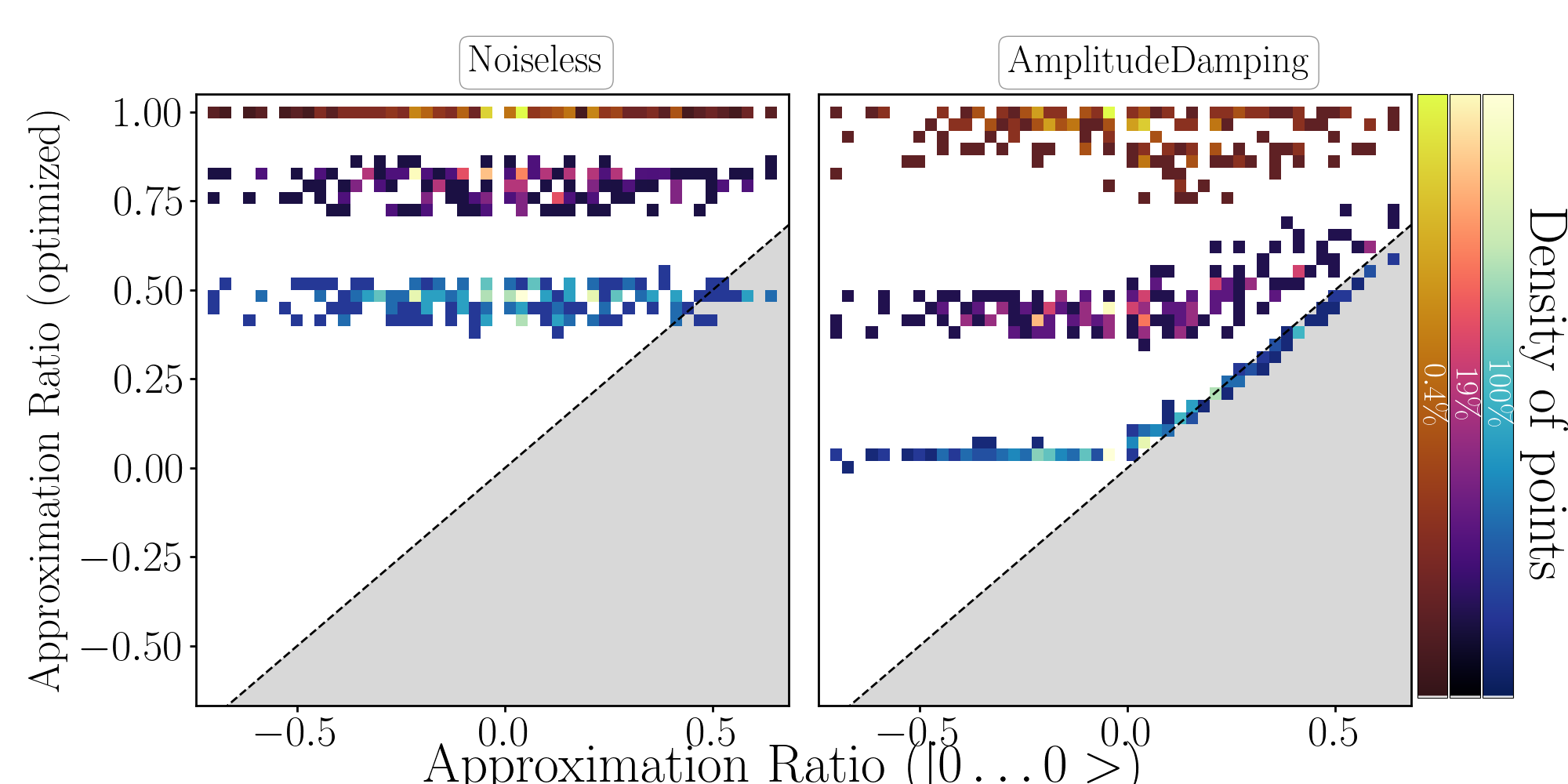}
\includegraphics[width=0.43\textwidth]{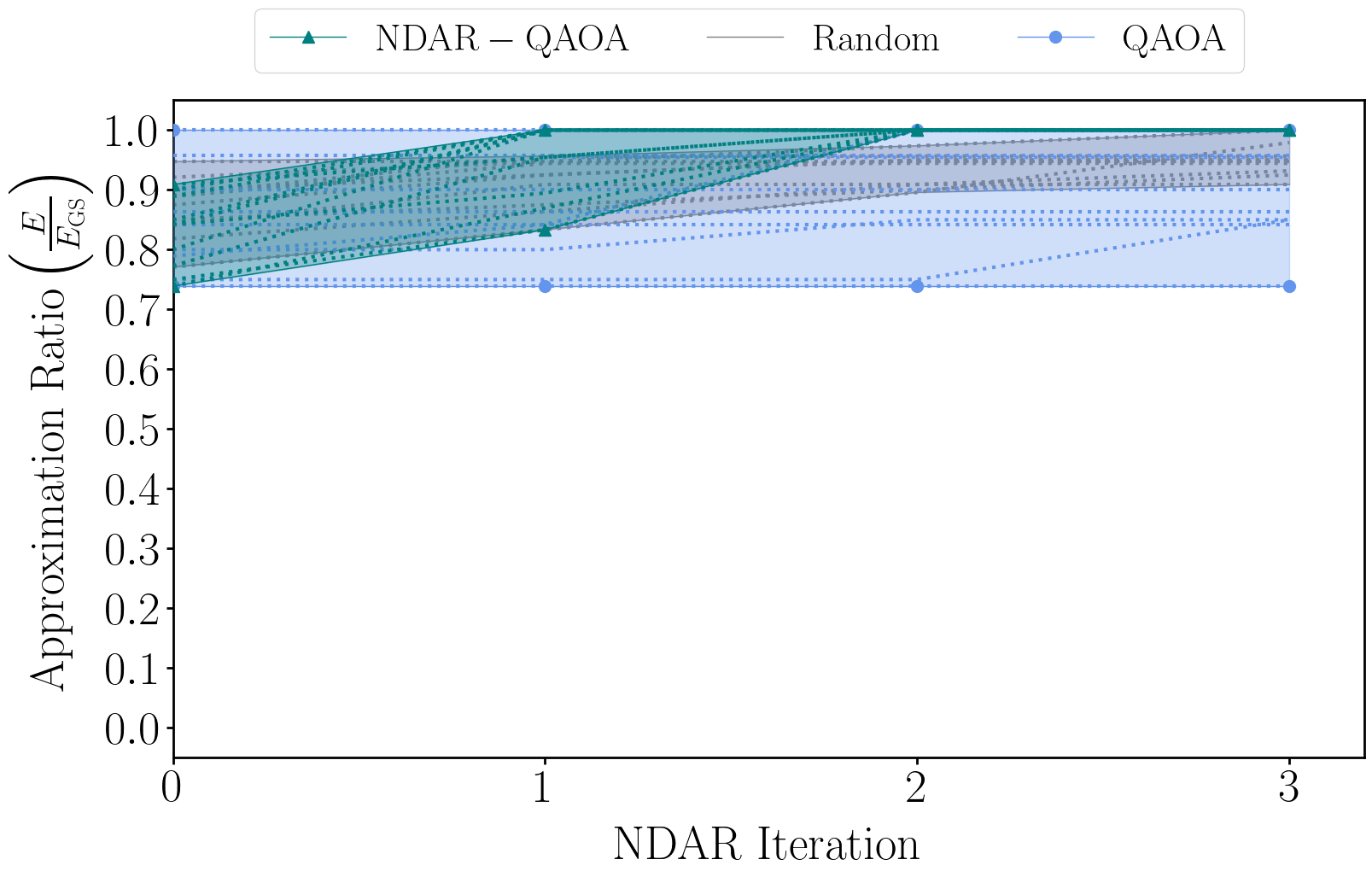}
\caption{
Investigation of the effects of applying bitflip gauges to improve an $n=16$-qubit variational optimization with $p=1$ QAOA in simulations.
The data presented is of the same type as Fig.~\ref{fig:exp:NDAR} in the main text, see description of that figure for details.
} \label{fig:sim:app:NDAR}
\end{figure*}

Recall that in Fig~\ref{fig:exp:NDAR} of the main text, we presented the results of two experiments on the Ankaa-2 device. 
The left side of Fig~\ref{fig:exp:NDAR}a presents a correlation investigation between the approximation ratio of the postulated attractor state (denoted as $\ket{0\dots 0}$) and QAOA's performance, while the right side of Fig.~\ref{fig:exp:NDAR} shows the overview of Noise-Directed Adaptive Remapping algorithm implementation.

We now consider the results of an analogous investigation of what we have presented in Fig.~\ref{fig:exp:NDAR} for simulated systems.
In our simulations, we use non-identical, uncorrelated strong amplitude-damping noise. 
A tensor-product, local two-qubit channel is applied to the circuit after every application of every gate. 
The 2-qubit gates are set to be more noisy than the 1-qubit gates.
We assume linear connectivity of a device and implement an efficient SWAP network from \cite{hirata2011} (same as in experiments in the main text).
This significantly increases circuit depth, making the simulations closer to what happens in experiments.
We assume that the native gate set consists of arbitrary ZZ and XY rotations.
All simulations are performed using Qsim framework \cite{qsim} that implements the quantum trajectories method \cite{isakov2021simulations} to approximate sampling from noisy circuits.

We implement $p=1$ QAOA optimization for $n=16$ qubits for the same class of Hamiltonians as in the main text (SK model).
The QAOA is implemented for a set of 10 random Hamiltonians. 
The aggregated results are presented in Fig.~\ref{fig:sim:app:NDAR}.

For the correlations plot (left side of Fig.~\ref{fig:sim:app:NDAR}), for each Hamiltonian, the optimization is run independently for $20$ random bitflip gauges, i.e., at a given optimization, the original Hamiltonian is gauge-transformed via a random bitfip $P_{\y}$ and such process is repeated independently for $20$   random $\y$s. 
Importantly, the optimizer seed is fixed in all simulations, meaning that the only sources of performance difference should be statistical noise, the gauge transformation, and the 
the varying hardness of optimization of particular Hamiltonian instances.
For each optimization run, we estimate the approximation ratio of the final output and compare the result with the corresponding approximation ratio.
The optimization is run with $t=100$ cost function evaluations, each involving gathering $s=256$ samples.
For reference, we plot also the results of noiseless optimization.

The plot on the left of Fig.~\ref{fig:sim:app:NDAR} shows, for the noiseless setting, no correlation between the AR of the attractor state and QAOA's performance.
This is expected since without noise all gauges are equivalent (recall discussion in Appendix~\ref{app:qaoa_symmetries}).
In the case of amplitude damping, however, we observe a positive correlation for mid-quantiles ($256$ best samples, and $50$ best samples, corresponding to $100\%$ and $\simeq20\%$ of all samples).
Interestingly, for poor gauge choices (AR below $0.0$), the performance is mostly flat for those quantiles, indicating that the optimizer gets stuck.
For the best-found solution (top $0.4\%$ of results), we do not observe a correlation between the two variables.
We expect this to be a small-size effect ($n=16$).
Indeed, note that the actual ground state (approximation ratio $1$) is found in multiple cases regardless of the gauge.

Due to the above observations, we decided to highly reduce the number of trials and samples for the NDAR implementation -- otherwise the NDAR would often converge at the first iteration (as indicated by results in Fig~\ref{fig:sim:app:NDAR}a), rendering the results uninteresting. 
The plot on the right of Fig.~\ref{fig:sim:app:NDAR} shows the performance overview of the Noise-Directed Adaptive Remapping algorithm applied to noisy simulated QAOA.
The QAOA is guided by the TPE optimizer with $t=20$ cost function evaluations, each involving gathering $s=100$ samples.
We observe that the NDAR applied to noisy QAOA (orange) finds the ground state in all tested cases after 3 iterations.
Results of standard QAOA (light blue) show that in that case, the performance is spread out across Hamiltonians, indicating that for some instances, the strong amplitude damping noise prohibits finding the ground state.
The reference random uniform baseline (grey), given the same small number of samples, achieves comparable performance. 
The above results indicate that NDAR can aid the noisy QAOA even when the computational resources (function calls and samples) are highly limited.

\section{Additional experimental data}\label{app:experiments_details}

\subsection{QAOA circuit implementation}\label{app:experiments_details:gates}

\begin{figure*}[!h]
    \centering
\includegraphics[width=0.59\textwidth]{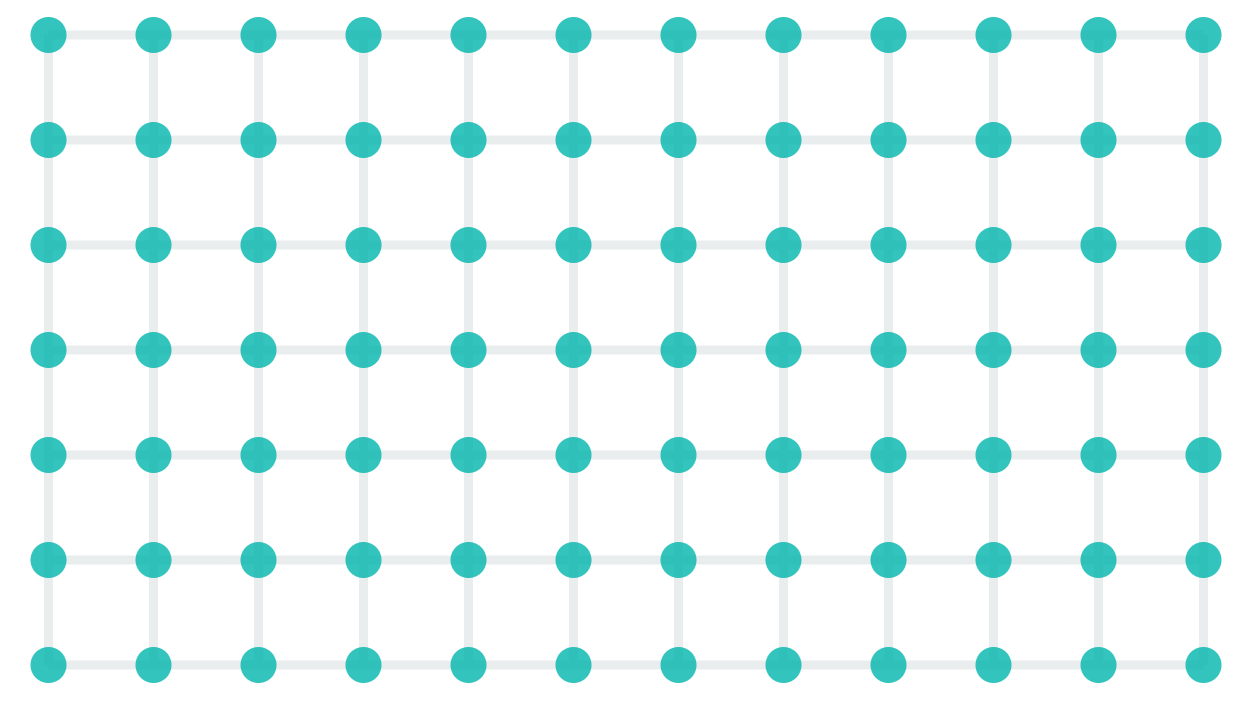}
\caption{\label{fig:exp:app:Anka}
Schematics of the Rigetti's Ankaa-2 superconducting transmon chip topology with respect to 2-qubit iSWAP operation. 
Missing edges not shown.
Source: rigetti.com.}
\end{figure*}

The Rigetti's Ankaa-2 (Fig.~\ref{fig:exp:app:Anka}) device's native gate set consists of iSWAP, $\pm \frac{\pi}{2}$ RX rotations and arbitrary  RZ rotations.
The device's connectivity (w.r.t. iSWAP gate) is that of a rectangular 7 by 12 grid (with some edges missing, not shown in Fig.~\ref{fig:exp:app:Anka}).
To implement QAOA circuits for fully connected Hamiltonians, we implement an efficient SWAP network from Ref.~\cite{hirata2011} on linear chains of qubits.
Since the phase separator (the $\exp\left(-i\gamma H\right)$ operator) of the QAOA ansatz implemented with SWAP network has a structure of multiple RZZ rotations followed by SWAP gates, we compile the product of the two gates jointly (a similar trick was applied for efficient compilation in Ref.~\cite{maciejewski2023design}).
Each implementation of RZZ+SWAP requires applying 3 ISWAP gates and at most 4 RX$\left(\pm\frac{\pi}{2}\right)$ and 5 RZ rotations (the exact number depends on the qubit).
Implementing a single layer of QAOA ansatz for $\noq$ qubits requires consecutive application of $\noq$ linear chains, each with $\approx \frac{\noq}{2}$ RZZ+SWAP gates.
Thus, for the biggest experiments shown in the text with $\noq = 82$, the $p=1$ QAOA ansatz consists of around $\approx 10,000$ ISWAP gates and $\approx 30,000$ single-qubit rotations.

\subsection{Parameter setting strategy  details}\label{app:experiments_details:optimization}

To implement variational optimization (for Figs.~\ref{fig:ndar_stairs_illustration}, \ref{fig:exp:NDAR}, and \ref{fig:exp:NDAR_histograms}) we use the Tree-Structured Parzen Estimators (TPE) implementation from the \texttt{optuna} package \cite{optuna2019}.
We use default hyperparameters (but with turned-off pruning trials). 
Each optimization was run with a random seed, and each optimization started from the same initial points (all angles set to $0.1$).

As indicated in the main text, besides 2 variational angles (for $p=1$ QAOA), we allow the optimizer to also vary between 10 random gates' orderings (GOs) of the ansatz (specified as categorical variables).
In the noiseless scenario, allowed GOs do not affect the circuit.
However, in real-life settings, it was demonstrated, for example in Ref.~\cite{maciejewski2023design}, that due to the noise, the optimization over GOs can highly improve performance.
When implementing the Noise-Directed Adaptive Remapping algorithm, to allow more versatility, we change the allowed set of GOs at random in each loop step.
To make a fair comparison when implementing standard QAOA (shown in Fig.~\ref{fig:exp:NDAR}b), we allow the optimizer to choose over the number of GOs equal to $10*i^{\max}$, where $i^{\max}$ is iteration index at which NDAR converged (for the corresponding Hamiltonian instance). 
This gives the TPE optimizer access to the same total number of GOs for QAOA optimization as for the NDAR implementation.
Note that standard QAOA is given also proportionally more function calls compared to a single NDAR step (i.e., $t=i^{\max}*150$).

As indicated in the main text, the experimental samples analysis for the 82-qubit system revealed that, interestingly, the bitstrings were concentrated close to $\ket{1\dots 1}$ physical state, as opposed to $\ket{0\dots 0}$ (the steady state of the amplitude damping channel). 
We suspect that this is due to overheating of the chip, but we intend to investigate the physical causes of observed effects in future work.
We note that since the Hamiltonians we implemented consisted of only $ZZ$ interactions, they exhibit $Z_2$ symmetry, meaning that the energy remains invariant under a global $X^{\otimes n}$ transformation.
This means that, for Hamiltonians considered in this work, $\ket{0\dots 0}$ has always the same energy as $\ket{1\dots 1}$. 
The consequence for the NDAR implementation is such that in practice the algorithm was not affected by whether we assumed that the physical attractor is $\ket{0\dots 0}$ or $\ket{1 \dots 1}$.

Whenever the approximation ratio is reported, we calculate $\frac{E}{E_{\mathrm{GS}}}$, where $E_{\mathrm{GS}}$ is the best-found solution obtained either via branch-and-bound SDP technique from Ref.~\cite{Baccari2020} or via simulated annealing package PySA ~\cite{PYSA2023}.



\end{document}